\documentclass[10pt,journal,compsoc]{IEEEtran}

\usepackage{setspace}
\usepackage[colorlinks=true,breaklinks=true,bookmarks=true,urlcolor=blue,
     citecolor=blue,linkcolor=blue,bookmarksopen=false,draft=false]{hyperref}

\usepackage[utf8]{inputenc}

\usepackage[american]{babel}
\usepackage{epsfig}
\usepackage{color}
\usepackage{makecell}
\usepackage{amsthm,amsfonts,amsmath,amssymb,amstext,latexsym}
\usepackage[linesnumbered,ruled]{algorithm2e} 
\usepackage{soul}
\usepackage{enumitem}
\setlist[itemize]{leftmargin=5.0mm}

\makeatletter
\def\endthebibliography{%
  \def\@noitemerr{\@latex@warning{Empty `thebibliography' environment}}%
  \endlist
}
\makeatother

\allowdisplaybreaks

\ifCLASSOPTIONcompsoc
  \usepackage[nocompress]{cite}
\else
  \usepackage{cite}
\fi
\ifCLASSINFOpdf
\else
\fi
\usepackage{amsthm,amsfonts,amsmath,amssymb,amstext,latexsym}

\ifCLASSOPTIONcompsoc
  \usepackage[caption=false,font=footnotesize,labelfont=sf,textfont=sf]{subfig}
\else
  \usepackage[caption=false,font=footnotesize]{subfig}
\fi

\usepackage{url}

\newtheorem{theorem}{Theorem}
\newtheorem{lemma}{Lemma}

\newtheorem{remark}{Remark}

\newtheorem{proposition}{Proposition}

\newtheorem{definition}{Definition}
\newtheorem{corollary}{Corollary}

\newtheorem{property}{Property}

\newtheorem{problem}{Problem}

\newcommand{\sd}{\pi}

\newcommand{\A}{\mathcal{A}}
\newcommand{\h}{\mathcal{H}}

\newcommand{\E}{\mathbb{E}}

\newcommand{\MBS}{{\rm ModifiedBS\mbox{-}\pi}}
\newcommand{\BS}{{\rm BS\mbox{-}\pi}}

\newcommand{\RBS}{R_{{\rm BS\mbox{-}\pi}}}

\newif\ifSOC
\SOCtrue

\ifSOC
    \newcommand{\review}[1]{{\color{black}#1}}
\else
    \newcommand{\review}[1]{{\color{black}#1}}
\fi

\begin{document}
%

\title{Balanced Splitting: A Framework for Achieving Zero-wait in the Multiserver-job Model}

%
%
%
%


\author{Jonatha~Anselmi and Josu~Doncel
\IEEEcompsocitemizethanks
{
\IEEEcompsocthanksitem J. Anselmi is with
Univ. Grenoble Alpes, CNRS, Inria, Grenoble INP, LIG, 38000 Grenoble, France.
E-mail: jonatha.anselmi@inria.fr

\IEEEcompsocthanksitem J. Doncel is with the University of the Basque Country, UPV/EHU, Barrio sarriena s/n, 48940 Leioa, Spain. 
E-mail: josu.doncel@ehu.eus
}
}

\IEEEpubid{}


\IEEEtitleabstractindextext{%
\begin{abstract}
We present a new framework for designing nonpreemptive and job-size oblivious scheduling policies in the multiserver-job queueing model. The main requirement is to identify a \emph{static and balanced sub-partition} of the server set and ensure that the servers in each set of that sub-partition can only handle jobs of a given \emph{class} and in a first-come first-served order. A job class is determined by the number of servers to which it has exclusive access during its entire execution and the probability distribution of its service time.
This approach aims to reduce delays by preventing small jobs from being blocked by larger ones that arrived first, and it is particularly beneficial when the job size variability intra resp. inter classes is small resp. large.
In this setting, we propose a new scheduling policy, Balanced-Splitting.
In our main results, we provide a sufficient condition for the stability of Balanced-Splitting
and show that the resulting queueing probability, i.e., the probability that an arriving job needs to wait for processing upon arrival, vanishes in both the subcritical (the load is kept fixed to a constant less than one) and critical (the load approaches one from below) many-server limiting regimes.
Crucial to our analysis is a connection with the M/GI/$s$/$s$ queue and Erlang's loss formula, which allows our analysis to rely on fundamental results from queueing theory.
Numerical simulations show that the proposed policy performs better than several preemptive/nonpreemptive size-aware/oblivious policies in various practical scenarios. This is also confirmed by simulations running on real traces from High Performance Computing (HPC) workloads.
The delays induced by Balanced-Splitting are also competitive with those induced by state-of-the-art policies such as First-Fit-SRPT and ServerFilling-SRPT, though our approach has the advantage of not requiring preemption, nor the knowledge of job sizes.
\end{abstract}




\begin{IEEEkeywords}
Multi-server jobs, online scheduling, queueing probability, zero-wait, Erlang's loss formula, asymptotic optimality
\end{IEEEkeywords}
}

\maketitle

\IEEEdisplaynontitleabstractindextext

%
\IEEEpeerreviewmaketitle

\IEEEraisesectionheading{\section{Introduction}\label{sec:introduction}}

%
%
%




\IEEEPARstart{T}{he multiserver-job}
queueing model has recently gained traction in the performance analysis of parallel systems because it captures a key aspect of today's High Performance Computing (HPC) and cloud computing systems \cite{MSQM22,MaxWeight,Grosof22SRPT,Danilo19SAF}.
Within this model, each job can occupy multiple servers simultaneously for the entire duration of the job itself.
Servers are an abstraction for processing resources:
in HPC systems, they refer to CPUs, cores, GPUs or nodes;
in cloud systems, they refer to virtual machines or containers; and
in serverless computing, they refer to cloud functions, instances or replicas.
The number of servers required by each job is referred to as \emph{server need} of that job,
and logs of real parallel workloads from production systems reveal that they range from one to a few thousand, depending on the application and on the resource type~\cite{feitelson2014experience,rep}.
In turn, these workloads run on large systems that are composed of up to millions of resources.
According to the TOP500 project \cite{TOP500},
the United States' Frontier is currently the most powerful supercomputer, reaching 8,699,904 total cores.

Recently, the interest for multiserver jobs has profoundly increased due to the massive use of deep learning technologies and applications, which rely on highly-parallel machine learning jobs like TensorFlow~\cite{TensorFlow} to handle several aspects of our daily life.
\review{Beyond the case of machine learning jobs, examples of jobs that typically require a specific number of servers are easily found
in linear algebra applications that rely on singular value or matrix decompositions,
in image processing scenarios such as neuroscience applications \cite{aupy}, where different servers operate on different sets of pixels, etc.
Finally, in \cite{tirmazi2020borg,Google_traces} the authors show that tasks submitted to Google's Borg Scheduler (which represent jobs executing internet-facing services such as Docs, Sheets, and Gmail, for instance, or jobs from the internal tools and services)
require a specific number of server, which can vary by five orders of magnitude across jobs (see also \cite{verma2015large}).}

In multiserver-job queueing models,
jobs wait in a central queue before being processed,
they then hold the simultaneous possession of multiple servers according to their server needs, and finally, they leave the system permanently upon service completion.
These models differ from standard queueing models where each job can only occupy one server, and they present the technical challenge of designing {scheduling policies} that can efficiently pack jobs onto servers ensuring full capacity usage whenever the number of jobs exceeds the number of servers. In this case, the corresponding scheduling policy is said \emph{throughput optimal}. However, it may be impossible to design a system that prevents a loss of capacity. For example, consider the case where the number of servers is odd but each job requires an even number of servers.

The fundamental problem is to design scheduling policies that yield high utilization/throughput and/or low delay/makespan. A scheduling policy for the multiserver-job model determines which jobs are processed by which servers at any given time.
This problem is nontrivial and referred to in the literature as the ``parallel on-line job scheduling problem'' \cite{Brucker}.
Since these policies may run on massively large computer systems, in practice HPC and cloud platform administrators opt to use simple heuristics, with low computational complexity.

In the last decades, several works have addressed the problem of designing scheduling policies for multiserver jobs.
Most existing studies focus their performance analysis on settings in finite time or with a finite number of jobs. In this paper, we consider an infinite stream of stochastically arriving jobs and investigate performance (stability and mean delays) in a \emph{steady-state} regime. The scientific literature has received relatively little attention in this queueing theoretic setting, which is attributed to the increased mathematical complexity caused by the non-conventional structure of multiserver jobs.
However, several new scheduling policies and interesting related results have recently been developed in the literature. Now, we provide an overview.

\subsection{Common scheduling policies and positioning}
\label{ssec:msj}

The list of scheduling policies for multiserver jobs investigated in the literature is huge~\cite{Brucker,Danilo19SAF,Patton,Grosof22SRPT} and
our goal here is to provide the necessary high-level background highlighting the difference of our work.
In general, it is difficult to establish whether a policy is better than another because the answer strongly depends on i) the assumptions made on the underlying architecture and ii) the information available to the scheduler.
In this respect, existing scheduling policies may be classified according to whether they require {preemption} or {job-size awareness}:
\begin{itemize}
 \item A scheduling policy is \emph{preemptive} if jobs can be stopped during their execution and resumed afterward at any point in time. While preemption may improve performance theoretically, it has some important limitations.  First, as discussed in~\cite{MaxWeight}, stopping and resuming jobs come with non-negligible switching costs, and job migration is also involved. Due to analytical intractability, existing queueing models do not usually take these costs into account, which leads to optimistic performance analyses. A further disadvantage of preemption is the starvation of jobs requiring a large number of servers.

 \item A scheduling policy is \emph{size aware} if at any point in time the scheduler knows the remaining processing time of each job together with its server need. Here, the celebrated example is the Shortest Remaining Processing Time scheduling discipline, which in the single-server job model is known to minimize delays in a strong sense~\cite{Schrage68}. The downside of this requirement is that in practice such information is rarely available. While users of HPC systems usually submit jobs together with estimations of their running times, it is well known that these are largely imprecise; see~\cite[Section~1.1]{LB_size_testing}
\end{itemize}

In turn, existing performance analyses of the multiserver-job queueing model may be classified according to whether they consider a \emph{transient} or \emph{steady-state} regime.
The former considers the scheduling of a finite number of jobs, while the latter considers an infinite stream of jobs entering the system at the jump times of an exogenous stochastic point process.
In this paper, we are interested in the steady-state regime.
In this setting, a primary performance objective concerns throughput optimality, which roughly means that ``the system is stable whenever the arrival rate is smaller than the service rate''.
Then, another important objective concerns the analysis of mean delay, or mean response time.
Unfortunately, tractable analytical formulas for mean delays are usually out of reach, even assuming that jobs arrive according to a Poisson process. For this reason, existing approaches investigate dynamics in some limiting regime of practical interest as this often leads to the identification of simple asymptotic formulas that turn out to be extremely accurate in the pre-limit.
In the literature, researchers have considered limiting regimes
where
i) the arrival rate approaches the boundary of the stability region from below while keeping constant the overall number of servers (\emph{heavy traffic} limiting regime)~\cite{Grosof22SRPT}, or
ii) the traffic demand grows to infinity together with the total number of servers (\emph{many servers} limiting regime)~\cite{Zwart,ZeroQueueing}.
In this paper, we are interested in the latter approach, as it fully captures the nature of real computer systems, which are composed of thousands, if not millions, of servers~\cite{TOP500}.

The requirements and features of the most common and newest scheduling policies proposed in the literature are summarized in Table~\ref{tab:summary_policies}; we point the reader to Section~\ref{sec:related_work} for more details about these scheduling policies.
In the list, we point out \emph{Balanced Splitting}, a scheduling policy that we propose in this work and discuss below.

\begin{table*}[t]
\begin{center}
\bgroup
\def\arraystretch{1.3}%
\begin{tabular}{lc|c||c|c}
\multicolumn{1}{c}{} & \multicolumn{2}{c}{\large{Requirements}} &  \multicolumn{2}{c}{\large{Performance}} \\
[-2ex]
\multicolumn{1}{c}{} & \multicolumn{2}{c}{\downbracefill}& \multicolumn{2}{c}{\downbracefill}\\
%
\multicolumn{1}{c}{} & Preemptive & Size-aware & \makecell{Throughput optimality}   &  \makecell{Many-server \\ delay asymptotics}  \\
\hline
\hline
\multicolumn{1}{l}{First-Come First-Served}  &  No  & No   & Almost$^*$ \cite{ZeroQueueing}  &  Yes \cite{ZeroQueueing}\\ \cline{2-5}
%
%
%
\multicolumn{1}{l}{Most Servers First \cite{BestFit,5493430}}  & Yes   &  No  & No  \cite{MaxWeight}  &  N/A \\ \cline{2-5}
%
\multicolumn{1}{l}{First-Fit Back-Filling}  & No   &  No  & ?  & N/A \\ \cline{2-5}
%
%
\multicolumn{1}{l}{First-Fit SRPT}  & Yes   &  Yes  & ?  &N/A \\ \cline{2-5}

\multicolumn{1}{l}{Max-Weight \cite{MaxWeight}}  &  Yes  & No   & Yes  &  N/A \\ \cline{2-5}
\multicolumn{1}{l}{GreedySRPT \cite{Danilo19SAF}}  &  Yes  & Yes   & ?  &N/A \\ \cline{2-5}
\multicolumn{1}{l}{ServerFilling \cite{Grosof22WCFS}} &  Yes  &  No  &  Yes$^{**}$ & N/A \\ \cline{2-5}
\multicolumn{1}{l}{ServerFilling-SRPT \cite{Grosof22SRPT}} &  Yes  &  Yes  &  Yes$^{**}$ & N/A \\ \cline{2-5}
\multicolumn{1}{l}{ServerFilling-Gittins \cite{Grosof22SRPT}} &  Yes  &  No  & Yes$^{**}$  & N/A \\ \cline{2-5}
%
\hline
\hline
\multicolumn{1}{l}{\makecell{\\\textbf{BalancedSplitting}-$\pi$ \\ (this paper)\\ \\}} & {{No}}   &  {No}  & {Almost$^{*}$}  &  {Yes}\\
[1.5ex]
\hline
\hline
\multicolumn{5}{l}{ $^{*}$: \emph{No} in general and \emph{yes} asymptotically in a many-server limiting regime, under mild conditions.} \\
\multicolumn{5}{l}{ $^{**}$: Provided that the overall number of servers and all the server needs are a power of two.} \\
\end{tabular}
\egroup
\end{center}
\caption{Summary of the requirements and performance of the most common and latest scheduling policies.}
\label{tab:summary_policies}
\end{table*}

\subsection{Balanced Splitting}
\label{ssec:contribution}

We introduce a new framework, referred to as Balanced Splitting Framework (BSF), for the design of nonpreemptive, job-size oblivious scheduling policies in the multiserver-job model.
The main idea is to reduce interference between small and large jobs by keeping them separate to some extent.
It works as follows.
First, divide jobs into \emph{classes}, where a job class is determined by its server need and its service time probability distribution.
Then, identify a static \emph{sub}-partition of the server set, say $\A:=\cup_i \A_i$, in proportion to the relative demand of each job class $i$, where static means independently of the job arrival process.
Finally, ensure that each subset of servers $\A_i$ can only process jobs of class $i$ and in a first-come first-served order. The servers outside $\A$ are called \emph{helpers} and their set is denoted by~$\h$.
%
Within BSF, we propose the following policy:
\vspace{0.12cm}
\begin{quote}
\textbf{BalancedSplitting}-$\pi$:
\emph{Upon arrival of a class-$i$ job, send it to servers in $\A_i$ if enough idle servers for immediate processing exist, otherwise send it to the helper set for \emph{potential} processing by the servers in~$\h$. In the helper set, process jobs according to  $\pi$, an auxiliary scheduling policy for multiserver jobs. Upon completion of service of a class-$i$ job in $\A_i$, assign to $\A_i$ the oldest class-$i$ job that is waiting for service in the helper set, provided that it exists.}

\end{quote}
\vspace{0.12cm}
While $\pi$ can be any of the policies discussed above, to preserve a simple structure we assume that it is nonpreemptive, size-oblivious, and only acting on the helpers, i.e., it operates independently of the state of the servers in~$\A$.

\subsection{Main contribution}

\review{Our main objective is to develop a scheduling policy yielding the so-called \emph{zero-wait property}, i.e.,
a condition where the mean job waiting time vanishes when the system size grows large together with the traffic demand.
In the queueing literature, this property has been considered in several works, e.g., \cite{Stolyar2015,Gamarnik2016,Anselmi20}, while for now it has received less attention in multiserver-job queueing models \cite{ZeroQueueing}; see Section~\ref{sec:related_work} for more details.
Within our approach, we address this objective by investigating the following question:
}
\vspace{0.12cm}
\begin{quote}
\emph{Within BalancedSplitting-$\pi$, under which conditions on the model parameters is it possible to neglect the probability, say $P_{\h}$, that a job is sent to the helper set~$\h$ upon its arrival?}
\end{quote}
\vspace{0.12cm}
\review{Here, the idea is the following}:
if $P_{\h}$ is ``small'', then the impact of the auxiliary policy $\pi$ is ``neglibible'' and the overall performance is captured by the dynamics that only occur in the $\A$ system.
This implies the {zero-wait property} because jobs that join the $\A$ system are processed immediately upon arrival.

We answer the question above in the affirmative within two limiting regimes of practical interest that are motivated by the huge size of real systems.
Specifically, in Theorems~\ref{th:subcritical} and~\ref{th:critical} we show that $P_{\h}$ vanishes
in both a \emph{subcritical} and a \emph{critical} many-server limiting regime where also the server needs are allowed to scale to infinity.
In the former, the demand and the overall number of servers grow to infinity in proportion while the load $\rho$ is kept fixed to a constant less than one.
In the latter, the difference is that at the same time the load $\rho$ approaches one from below.
Both regimes have been widely considered in the queueing literature \cite{Kelly91Loss,HW81,Hunt1989}, and in our case they are well justified because of the massive size of modern computer systems.
The critical many-server regime considered in this paper corresponds to the Halfin-Whitt regime, which has been mainstream in the queueing literature since the seminal work in~\cite{HW81}.
This regime aims to balance between efficiency and quality of offered service and is also known as Quality and Efficiency Driven (QED).

At the core of our proofs, there is connection a connection with the M/GI/$s$/$s$ queue, which is a well-studied object in queueing theory since the pioneering works of Erlang.
More specifically, we show that the dynamics induced by BalancedSplitting-$\pi$ can be bounded by the dynamics of a modified version of BalancedSplitting-$\pi$ ensuring that the~$\A$ system can be decomposed into multiple independent M/GI/$s$/$s$ queues.
Within this decomposition, the $\A$ system becomes
tractable from a mathematical point of view and
the question above can be essentially rephrased as
``\emph{Under which conditions is it possible to make the blocking probability of an M/GI/$s$/$s$ queue vanish?}''.
This question has a long history in queueing theory since the work of Erlang (see, e.g., \cite{Zwart08Erlang}), and here we use existing theorems to establish our results.

Figure~\ref{fig:1} plots the simulated mean response time induced by BalancedSplitting-$\pi$ and other scheduling policies in the critical many-server limiting regime.
Here, the server needs grow with rate $f_k= \lfloor (k/32)^{2/3} \rfloor$ where $k$ is the overall number of servers, and the load $\rho=\rho(k)$ satisfies $(1-\rho)\sqrt{k/f_k}\to\theta=0.7$ as $k\to\infty$.
Job sizes are \emph{highly variable} and follow the model of ``several small and few large'' jobs, which is typical in HPC and cloud systems.
The relative demand of large resp. small jobs is $80\cdot 0.05f_k =4 f_k$ resp.  $0.95f_k$; see the figure caption for details.
For BalancedSplitting-$\pi$, we assume that the auxiliary policy~$\pi$ is First-Come First-Served (FCFS).
\begin{figure}
\centering
\makebox[1.0\columnwidth][c]{\hspace{0.3cm}\includegraphics[height=7.9cm]{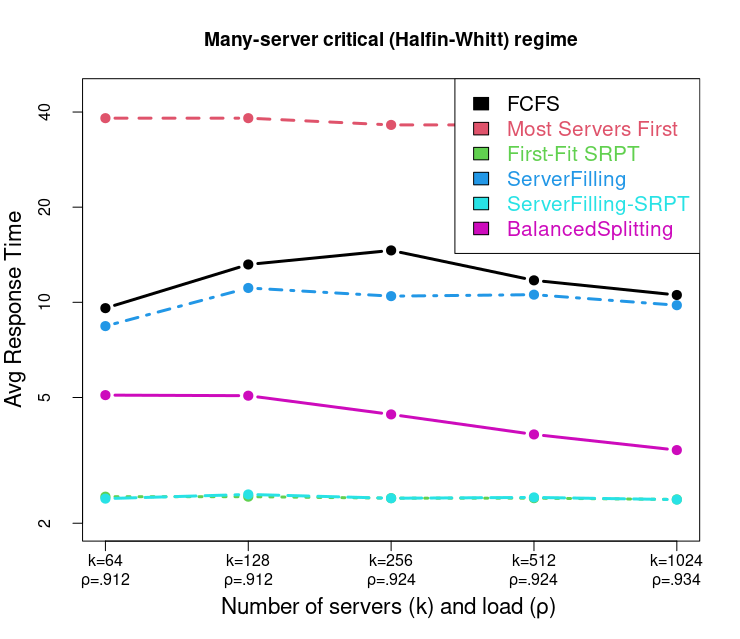}}%
%
%
\caption{
Simulated mean response time in the critically loaded limiting regime (Halfin--Whitt).
Small resp. large jobs arrive with probability 0.95 resp. 0.05.
Small jobs have server need and mean service time equal to $(f_k,1)$, where
$f_k= \lfloor (k/32)^{2/3} \rfloor$.
Large jobs have server need and mean service time equal to (2$f_k$,40), (4$f_k$,20) or (8$f_k$,10), with equal probability.
Service times are exponentially distributed and all independent.
Arrivals are Poisson.
Each simulation run uses $10^6$ arrivals.
The curves associated to ServerFilling-SRPT and First-Fit SRPT almost overlap.
}
%
\label{fig:1}
\end{figure}
First, we notice that BalancedSplitting-$\pi$  outperforms FCFS.
It also performs better than ServerFilling, which is preemptive, and provides average response times that are competitive with those achieved by ServerFilling-SRPT and First-Fit SRPT, but without incurring the cost of preemption, nor requiring the knowledge of job sizes.
\review{
In Section~\ref{sec:numerical}, we will complement our analysis by presenting an additional set of numerical simulations against both synthetic and real data traces. These will provide further evidence of the benefits of our approach and scheduling policy.
}


\section{Scheduling policies review}
\label{sec:related_work}

This section complements Section~\ref{ssec:msj} and reviews in more detail the most common and latest scheduling policies proposed in the literature (see also Table~\ref{tab:summary_policies}).
Several of these are considered in this paper for benchmarking purposes.

\begin{itemize}

\item  FCFS: First-Come First-Served. Jobs are processed in order of arrival if there are enough servers to meet their needs, otherwise they wait in the central queue. A queueing theoretic analysis of the resulting multiserver-job queueing model is challenging even assuming exponentially-distributed job durations and Poisson arrivals~\cite{MSQM22,ZeroQueueing}.

\item Most Servers First.
At all times, jobs with the highest server need are processed first among all jobs that can be served, with preemption
\cite{5493430,BestFit}. This policy is also known as Best-Fit~\cite{MaxWeight}.
%

\item \review{Least Servers First. At all times, priority is given to jobs with the smallest server needs~\cite{Grosof22WCFS}.
Within the same parameter settings used for Figures~\ref{fig:1},~\ref{fig:2} and \ref{fig:KIT}, our simulations indicate that the response times induced by Least Servers First are huge. In fact, this policy does not seem stable within the considered range of parameters. For this reason, in the remainder we omit in our plots the corresponding numerical results.
}

\item  First-Fit Back-Filling. As FCFS, but if there are idle servers and jobs waiting in the queue that can be executed, then the first arrived job that fits the idle servers is processed.
Several variants of this ``Back-Filling'' type of policy are available in the literature \cite{1039773,Danilo19SAF}.

\item  First-Fit Shortest Remaining Processing Time (SRPT).
Serve the jobs with the least remaining processing time, regardless of their server needs. If a job has a higher server need than the remaining number of servers available, skip that job and continue through the list of jobs, placing jobs into service if sufficient servers are available, until all servers are full, or all jobs are assigned. This policy resembles policies like ``Smallest Area First'' and EASY BackFilling~\cite{Danilo19SAF}.




\item  MaxWeight. At all times, choose a packing $x$ that maximizes $\sum_n Q_n x_n$ where $Q_n$ is the number of jobs with server need~$n$ that are in the system and~$x_n$ is the number of jobs with server need $n$ that are served by packing~$x$~\cite{MaxWeight}.
While MaxWeight is throughput optimal, it is not practical as it requires the scheduler to solve an NP-hard optimization problem any time a job arrives or departs.
There is a myopic version MaxWeight that is nonpreemptive and allocates a new job to a server using current queue length information at job departure times. Under some technical assumptions, it can achieve any arbitrary fraction of the stability region~\cite{MaxWeight}.

\item  ServerFilling. Identify the minimum set of jobs $M$ that collectively require a number of servers at least equal to~$k$, in order of arrival. If such a set is empty, all jobs present in the system are served simultaneously (if any exist), otherwise, the jobs in $M$ are ordered by their server needs, from largest to smallest (with ties broken by arrival order), and placed into service in that order until no more servers are available. The overall number of servers $k$ and server needs do not need to be a power of two. If $k$ is not a power of two, a job $j$ in $M$ may not be packed in even when idle servers exist, and in this case the scheduling is completed with all the jobs arrived later than $j$ blocked and some servers unused. ServerFilling has been proposed in  \cite{Grosof22WCFS}.

\item  ServerFilling-SRPT \cite{Grosof22SRPT}. Order jobs in increasing order of remaining size $r_j$, breaking ties arbitrarily, so that $r_{j_1}\le r_{j_2},\ldots$ for jobs $j_1$, $j_2$, etc. Here, the size of a job is its remaining service time times its service need. Let $m^*$ be the smallest integer such that $$
    \sum_{m=1}^{m^*} n_{j_m} \ge k,
    $$ where $n_{j_m}$ is the server need of job $j_m$ If there is no such index, then serve all jobs in the system simultaneously (if any exists). Otherwise, serve the jobs in $M:=\{j_1,j_2,\ldots,j_{m^*}\}$ prioritizing jobs of largest server need until no servers remain or the next job cannot fit, breaking ties by smallest remaining size.

\item  ServerFilling-Gittins \cite{Grosof22SRPT}. Identical to ServerFilling-SRPT except that jobs are ordered in increasing order of \emph{Gittins rank}, which for a job with server need $i$, service time $S_i$ and age (time spent in the system) $a$ is defined as
$
\inf_{b>a}  {\mathbb{E}[\min(S_i,b)-a\mid S_i> a] / P(S_i\le b\mid S_i> a)}
$.
If the service times of jobs with a given server need are deterministic or if job sizes are known, then ServerFilling-Gittins and ServerFilling-SRPT are equivalent.

\end{itemize}


Most of the existing policies require preemption. As discussed above, this is costly in practice.

Throughput optimality for the recent ServerFilling-* policies holds if the overall number of servers $k$ and all the server needs are a power of two. Within this assumption, it is possible to perfectly pack jobs into the server set whenever the overall amount of server need exceeds the number of servers.
The family of DivisorFilling-* policies discussed in \cite{Grosof22SRPT} weakens this assumption requiring that all server needs are a divisor of~$k$. Their applicability remains limited because in practice the overall number of servers is hardly a power of two~\cite{TOP500}. One of our goals is to develop a scheduling policy that is more flexible in this respect.

\subsection{Queueing probability and the zero-wait property}

In the queueing literature, many-server limiting regimes have been extensively considered to evaluate the queueing probability, i.e., the probability that an arriving job needs to wait upon arrival.
In several cases, this probability converges to zero in the limit, yielding the zero-wait property~\cite{HW81,Kelly91Loss,Hunt1989,Zwart08Erlang}.
In the multiserver-job queueing model however, only a few results are available.
In~\cite{ZeroQueueing}, the authors analyze the FCFS scheduling policy under the assumption that service times follow an exponential distribution.
While they prove that FCFS is not throughput optimal, the suboptimality gap vanishes in a many-server limiting regime and the zero-wait property is achieved. In this paper, we follow a similar approach to theirs, but with respect to our new policy BalancedSplitting-$\pi$.

\section{Framework}
\label{sec:model}

This section contains our contribution from a modeling and algorithmic point of view.
First, it describes the multiserver-job queueing model as in~\cite{ZeroQueueing}. Then, it defines the proposed Balanced Splitting framework introducing two scheduling policies and related performance indicators.

\subsection{Multiserver-job queueing model}

We consider a queueing system where an infinite stream of jobs needs to be processed by~$k$ parallel servers working with unitary speed.
Jobs join the system following a Poisson process with rate~$\lambda$ and can be partitioned into $C$~\emph{classes}. A job is of class~$i$ with probability $\alpha_i$, for all $i=1,\ldots,C$.
Class-$i$ jobs require the simultaneous possession of $n_i\le k$ servers for a random amount of time equal in distribution to the random variable~$D_i$, for all $i=1,\ldots,C$.
After processing, they leave the system.
In the following, $n_i$ (a constant) and $D_i$ will be respectively referred to as ``server need'' and ``service time'' of class-$i$ jobs.
%
%
The stochastic sequences of job inter-arrivals, service times, and classes are assumed i.i.d. and independent.

Let $d_{i}:=\E[D_i]$. We denote by $\varrho_{i} :=  \alpha_i d_{i} n_i$ the \emph{relative demand} induced by class-$i$ jobs, and let $\varrho :=  \sum_{i=1}^C  \varrho_{i}$.
We assume that the \emph{load}
\begin{equation}
\label{eq:stab_cond}
\rho := \frac{\lambda}{k} \varrho < 1,
\end{equation}
which is necessary to ensure stability (positive Harris recurrence) of the underlying Markov process.
Note that \eqref{eq:stab_cond} corresponds to the usual stability condition ``the overall arrival rate is less than the overall service rate''.


%

Within the model described above, a \emph{scheduling policy} is a rule that decides which jobs are processed by the servers at any point in time.

\subsection{Partitioning the set of servers}
\label{sec:partitioning}
For a set $S$, let $|S|$ denote its cardinality.
We let sets $(\A_{i})_{i}$ and $\h$ form a partition of the set of servers $\{1,\ldots,k\}$ such that for all~$i$
\begin{subequations}
\label{eq:A_and_B}
\begin{align}
\label{eq:n_ij}
|\A_{i}| & =\left \lfloor \psi\, \frac{k}{ n_i} \,\frac{\varrho_{i}}{\varrho}  \right\rfloor n_i = : a_{i} \\
|\h| & =  k-\sum_{i=1}^C a_{i}
\end{align}
\end{subequations}
where $\psi=1$ if $\frac{k}{ n_i} \,\frac{\varrho_{i}}{\varrho}$ is an integer for all $i$,
and otherwise $\psi = \max\{x\in[0,1]: k-\sum_i \lfloor x\, \frac{k}{ n_i} \,\frac{\varrho_{i}}{\varrho}  \rfloor n_i \ge \max_i n_i\}$.

The servers in~$\A_{i}$ will only process jobs of class $i$, while the servers in $\h$ will process jobs of potentially any class and are interpreted as \emph{helpers}.
We refer to $\cup_{i} \A_{i}$ as the $\A$ system.

The parameter~$\psi$ ensures that the helper set contains enough servers to serve jobs of any class, provided that it is non-empty.
To minimize the interference between different job classes, and therefore to fully exploit the benefits of our approach,
we want the helper set to be small compared to $\A$.
To achieve this, it should be clear that $\psi$ should be as close as possible to one. However, the choice~$\psi=1$ may not always be convenient because the resulting number of helpers may be less than some job's server need. Hence, the choice of $\psi$ given above.


In the following, our analysis assumes that $a_{i}>0$ for all $i$; if this is not the case, one can always redirect class-$i$ jobs to the helpers.

Finally, we notice that the helper set is empty only when $\frac{k}{ n_i} \,\frac{\varrho_{i}}{\varrho}$ is an integer for all $i$.
This condition has ``zero measure'' because $\varrho_i$ is a real number. To avoid unnecessary and minor complications, e.g., divisions by zero, in the remainder we assume that $\h$ is non-empty ($|\h|\ge\max_i n_i$). If $\h$ would be empty, all of our results remain true trivially.

\subsection{Balanced Splitting framework}
\label{sec:BS}

The Balanced Splitting Framework (BSF) specifies a principle for scheduling policies.
Namely, a scheduling policy belongs to BSF if the subset of servers $\A_{i}$, defined in Section~\ref{sec:partitioning}, process jobs of class~$i$ only and according to FCFS.
In this manner, job class information is exploited to reduce the interference between jobs with different server need requirements, as their processing is forced to occur on different sets of servers.

Inside BSF, we propose the scheduling policy \emph{Balanced Splitting}-$\pi$, denoted by $\BS$, defined as follows.

\begin{definition}
\label{def:BS}
%
$\BS$ is the BSF policy that is based on the following rules:
\begin{enumerate}
 \item
Upon the arrival of a class-$i$ job, send it to servers in $\A_i$ if enough idle servers for immediate processing exist, otherwise send it to the helper set for \emph{potential} processing by the servers in~$\h$.

\item
In the helper set, process jobs according to $\pi$, an auxiliary scheduling policy for multiserver jobs. The policy $\pi$ is nonpreemptive, size-oblivious, and operates independently of the~$\A$ system.

\item
Upon completion of service of a class-$i$ job in $\A_i$, assign to $\A_i$ the oldest class-$i$ job that is waiting for service in the helper set, provided that it exists.
\end{enumerate}

\end{definition}


Several other policies can comply with the BSF principle above.
First, $\pi$ may be preemptive and/or size-oblivious.
Then, point (iii) in Definition~\ref{def:BS} may be tweaked to assign to the servers in $\A_i$ the class-$i$ job in the helper set with the highest remaining size or with the highest \emph{expected} remaining size, as this would reduce the workload on the helper set the most.
We do not consider these size-aware variants because we want to keep the overall structure as simple as possible.

The following remark summarizes the minimal requirements needed by $\BS$.
\begin{remark}
$\BS$ is nonpreemptive and size-oblivious.
The scheduler only needs to know the class of each arriving job and, to build the $\A$ system,  the overall number of servers together with the relative demands $\varrho_i$, which can be estimated a priori by user-profiling techniques.
In particular, it does not need to know the arrival rate of jobs.
\end{remark}

%
%
%

\subsection{Problem statement}
\label{sec:performance_measures}

Given a scheduling policy $\sd$, $R_{\sd}$ denotes the mean response time (or delay) induced by~$\sd$, i.e., the mean time spent in the system by a job.
In this paper, we are mainly interested in analyzing the mean response time induced by BS-$\pi$.
This comes with the technical difficulty of choosing the auxiliary policy $\pi$, which in principle is ``arbitrary'' and would require a separate analysis; for $\pi$ generic, analytical formulas for $R_{\pi}$ are not available.
To overcome this difficulty,
it is then convenient to introduce the steady-state probability that a job, upon arrival, needs to use the servers in the helper set.
Let $P_{\h}$ denote such probability.
If $P_{\h}$ is ``small'', then the contribution of $\pi$ on the mean response time induced by $\BS$ is ``small''.
This motivates us to investigate the following problem.

\begin{problem}
\label{problem:1}
We are interested in
understanding under which  conditions
it is possible to drive $P_{\h}$ to zero.
\end{problem}

\section{Main results}
\label{sec:main}

This section presents our main results about the performance induced by the proposed policy $\BS$.
\review{The key observation, which will be used in our proofs, is that
a class-$i$ job is processed either by servers in~$\A_{i}$ or by servers in~$\h$, disjointedly.
This holds true because the number of servers dedicated to class-$i$ jobs is a multiple of $n_i$, by construction of the set of servers $\A_{i}$.}
%
This property will be exploited to show that
the mean number of busy servers in the $\A_i$ system is lower bounded
by the mean number of busy servers of an independent M/GI/$s$/$s$ queue; see Section~\ref{sec:connection_MGss}.
In light of this connection, we will present our results using standard performance indicators that are relevant to this queue, for which we recall some of its main properties before stating our results.

\subsection{Preliminaries}

The M/GI/$s$/$s$ queue is a crucial component of queueing theory and the investigation of the \emph{blocking probability}, i.e., the probability that an incoming job finds all servers busy, is undoubtedly one of the most studied subjects in the literature.
Assuming that the arrival process is Poisson with rate $\lambda$ and that the mean service time of each server is $d$, it is well known that the blocking probability, which we denote by $E_s(\lambda\, d)$, is given by the Erlang's loss formula:
\begin{align}
\label{eq:ErlangB}
E_s(\lambda\, d) = \frac{(\lambda\, d)^s}{s!}  \left(\sum_{\ell=0}^s\frac{(\lambda\, d)^\ell}{\ell!}\right)^{-1}.
\end{align}
This is connected to the mean response time, denoted by $R_s$, through the relation
\begin{align}
\label{eq:ErlangB_R}
R_s
= d \,(1-E_s(\lambda\, d)).
\end{align}
Also, we recall the following classical result due to Erlang~\cite{Zwart08Erlang}.
\begin{lemma}
\label{Erlang_lemma}
Suppose that the arrival rate $\lambda$ and the number of servers $s$ in an M/GI/$s$/$s$ queue grow to infinity while the load (or utilization) $\rho=\lambda d/s $ approaches one and
$(1-\rho)\sqrt{s}\to\theta$
where $\theta$ is a ﬁxed constant. Then,
$$
\lim_{s\to\infty}
\sqrt{s} E_{s}(\lambda d) = \frac{\phi(\theta)}{\Phi(\theta)},
$$
where $\Phi(x)$ and $\phi(x)$ denote the standard normal cumulative distribution function (CDF) and density, respectively.
\end{lemma}

The scaling in Lemma~\ref{Erlang_lemma} satisfies the conditions that define the Halfin-Whitt limiting regime~\cite{HW81}.

%

\subsection{Stability}

We say that the scheduling policy $\sd$ is \emph{throughput optimal} if $R_{\sd} <\infty$
whenever the load $\rho$ is less than one.

Our first result provides a sufficient condition for stability; see Section~\ref{sec:proofs} for a proof.
\begin{proposition}
\label{pr1}
Assume that $\pi$ is throughput optimal.
If
\begin{align}
\label{eq:cond_TO_BS}
\frac{\lambda}{|\h|} \sum_{i=1}^C \varrho_i  E_{s_i}(\lambda \alpha_i d_i) <1,
\end{align}
then $\RBS<\infty$.
\end{proposition}

It is not clear whether $\BS$ is throughput optimal (when~$\pi$ is throughput optimal on the helper set).
The problem is that service capacity can be lost as follows.
Suppose that $(x_1,\ldots,x_I,h_1,\ldots,h_C)$ is a ``state'' of the system where~$x_i$ resp.~$h_i$ denotes the number of class-$i$ jobs in $\A_i$ resp. in $\h$. For simplicity, suppose also that there are only two classes, i.e., $C=2$, and that $n_1<n_2$.
%
Then, a loss of capacity occurs in state $(\frac{a_1}{n_1},0,h_1,0)$ with $h_1>\lfloor\frac{|\h|}{n_1}\rfloor$, as one class-1 job that is waiting for service in the helper set may be processed by the servers in~$\A_2$.
This does not occur due to the static server set partitioning of the proposed approach,
which forces jobs of different classes not to interfere with each other.
%
%

In the following, we show that \eqref{eq:cond_TO_BS} is eventually true if the number of servers is large enough. In other words, this means that $\BS$ is throughput optimal \emph{asymptotically}.

\subsection{Subcritically loaded many-server regime}

Now, we consider a limiting regime where the demand $\lambda \varrho$ and the overall number of servers $k$ grow to infinity in proportion while the load $\rho$ is kept fixed to a constant less than one.
In the multiserver-job model, there are several ways to accomplish this as both the arrival rate $\lambda$ and needs~$n_i$, for all $i$, can grow with $k$.
To identify a proper scaling of the system parameters, we consider a sequence of models indexed by $k$ and add the superindex $(k)$ to refer to the $k$-th system.

First, we introduce the sequence $(f_k)_k$ such that
\begin{align}
\label{eq:cond_f_k}
f_k =o(k),\quad f_k\in\mathbb{N}.
\end{align}
Then, we let $k\to\infty$ with the following scaling:
\begin{subequations}
\label{eq:LSL_cond}
\begin{align}
\lambda^{(k)} & =\lambda \frac{k}{f_k}\\
n_i^{(k)} & =n_i\, f_k\\
\alpha_i^{(k)} & =\alpha_i, \quad D_i^{(k)} =D_i,
\end{align}
\end{subequations}
which ensures that the network load $\rho^{(k)}$ remains constant to $\rho=\lambda\varrho<1$.
This subcritical regime has been widely considered in the queueing literature; e.g.,~\cite{Kelly91Loss}.

Within this regime, the next result shows that the probability that
an arriving job is routed to the helper set, $P_{\h}^{(k)}$, converges to zero; see Section~\ref{sec:proofs} for a proof.
\begin{theorem}
\label{th:subcritical}
Let \eqref{eq:cond_f_k} and \eqref{eq:LSL_cond} hold.
As $k\to\infty$,
$P_{\h}^{(k)}\to 0$ and $\RBS^{(k)}\to \sum_i \alpha_i d_i$.
\end{theorem}

\subsection{Critically loaded many-server regime}

Proceeding in a similar way as above,
we now consider a limiting regime where the demand $\lambda \varrho$ and the overall number of servers $k$ grow to infinity in proportion but at the same time, the load $\rho$ approaches one from below.
Here, let the sequence $(f_k)_k$ satisfy again the conditions in \eqref{eq:cond_f_k}
and consider the following scaling:
\begin{subequations}
\label{eq:HW_cond}
\begin{align}
\lambda^{(k)} & \to \infty\\
\left(1-\rho^{(k)} \right) \sqrt{\frac{k}{f_k}} &\to \theta,  \quad \theta>0 \\
n_i^{(k)} & =n_i\, f_k\\
\alpha_i^{(k)} & =\alpha_i, \quad D_i^{(k)} =D_i.
\end{align}
\end{subequations}
This scaling satisfies the conditions that define the Halfin-Whitt regime \cite{HW81}.
Within this scaling and since $f_k=o(k)$, the following result implies that the probability that a job is routed to the helper set $P_{\h}$ vanishes as $k\to\infty$, and it also provides a convergence speed; see Section~\ref{sec:proofs} for a proof, which is based on Lemma~\ref{Erlang_lemma}.
\begin{theorem}
\label{th:critical}
Let~\eqref{eq:cond_f_k} and~\eqref{eq:HW_cond} hold.
As $k\to\infty$,
\begin{align}
\label{th:P_h}
\lim_{k\to\infty} \sqrt{\frac{k}{f_k}}\,P_{\h}^{(k)}
& \le \theta  \sum_{i=1}^C \frac{\alpha_i}{\theta_i} \frac{\phi(\theta_i)}{\Phi(\theta_i)},\quad
\theta_i := \theta  \sqrt{\frac{\varrho_{i}}{n_i\varrho}}
\end{align}
where $\Phi(x)$ and $\phi(x)$ denote the standard normal cumulative distribution function (CDF) and density, respectively.
\end{theorem}

\section{Proofs of our main results}
\label{sec:proofs}

\review{In this section, we provide proofs for our main results along with the intuition behind our strategy.
Roughly speaking, the general strategy for proving our results is as follows.
Since the dynamics underlying~$\BS$ are presumably intractable to investigate due to the complex interactions between the $\A$ and $\h$ systems,
we first define a \emph{modified} version of $\BS$ where the $\A$ and $\h$ systems are independent. Here, we show that the $\A$ system behaves exactly like a set of parallel M/GI/$s$/$s$ queues.
Then, we connect this modified version to $\BS$ by showing that it provides bounds on the performance metrics of interest.
Finally, we rely on these bounds to develop our asymptotic results.}

\subsection{Modified policy and connection with the M/GI/$s$/$s$ queue}
\label{sec:connection_MGss}

Let us define the following modified version of $\BS$, denoted by $\MBS$.

\begin{definition}
\label{def:MBS}
$\MBS$ is the BSF policy where
an arriving job of class $i$ is sent to $\A_{i}$ if $\A_{i}$ contains enough idle servers for immediate processing,
otherwise it is \emph{irrevocably} sent to the helper set, for all $i$.
The helper set processes jobs according to an auxiliary scheduling policy $\pi$ for multiserver jobs that operates independently of the state of the $\A$ system.
\end{definition}

Note that $\MBS$ is a scheduling policy that satisfies the BSF principle.
The main difference with respect to $\BS$ is that jobs that upon arrival are sent to the helper set will be certainly processed by the helpers, while within $\BS$ they may be reassigned to the $\A$ system.
For this reason, one expects that it performs worse than $\BS$.
From a mathematical point of view, the advantage of $\MBS$ over $\BS$ is that the dynamics of the $\A$ system are autonomous and therefore they can be analyzed independently of the $\h$ system -- this is not true within $\BS$ because arrivals to the $\A$ system depend on the state of the helper set.

\review{
At this point, the key observation for the dynamics induced by $\MBS$ is that
a class-$i$ job is processed either by servers in~$\A_{i}$ or by servers in~$\h$. In other words, it is not possible that a job is processed by servers in~$\h$, preempted and then processed by servers in~$\A_{i}$.
In addition, the decision whether a job is assigned to the $\A$ or $\h$ system is irrevocable and made upon job arrival.
In view of this, the system enjoys the following property.}

\begin{property}
\label{property3}
Within $\MBS$, the process of the number of servers occupied in each $\A_{i}$ subsystem divided by $n_i$ is the process of the number in an M/GI/$s_{i}$/$s_{i}$ queue
where
$s_{i}
= \lfloor \psi \frac{k \varrho_{i}}{n_i\varrho} \rfloor$,
the arrival rate is $\lambda \alpha_i$
and service times are distributed as~$D_i$.
\end{property}

Now, it remains to establish the connection between the performance induced by $\BS$ and $\MBS$. This is stated in the following proposition.
\begin{proposition}
\label{prop:asi9c8s}
Let $\lambda_{i}^{\BS}$ and $\lambda_{i}^{\MBS}$ denote the arrival rates of class-$i$ jobs at the $\A$ system induced by $\BS$ and $\MBS$, respectively.
Then,
\begin{equation}
\lambda_{i}^{\BS} \ge \lambda_{i}^{\MBS}.
\end{equation}
\end{proposition}
\begin{proof}
Within $\MBS$, the arrival rate of class-$i$ jobs at the $\A$ system is $\lambda \alpha_i$ minus the rate of blocked jobs of the M/GI/$s_i$/$s_i$ queue defined in Property~\ref{property3}.
Therefore,
\begin{align}
\lambda_{i}^{\MBS} = (1-E_{s_i}(\lambda \alpha_i d_i)) \lambda \alpha_i.
\end{align}
For~$\lambda_{i}^{\BS}$, we notice that $E_{s_i}(\lambda \alpha_i d_i)$ is higher than or equal to the probability that an incoming job will be processed by the helper set as $\BS$ may redirect that job to $n_i$ servers in~$\A_i$.
\end{proof}

The following corollary of Proposition~\ref{prop:asi9c8s} is then straightforward.
\begin{corollary}
\label{cor:asi9c8s}
Let $P_{\h}^{\MBS}$ denote the
probability that a job, upon arrival and within $\MBS$, needs to use the servers in the helper set.
Then,
\begin{equation}
P_{\h} \le P_{\h}^{\MBS}.
\end{equation}
\end{corollary}

\subsection{Proof of Proposition~\ref{pr1}}

In view of Property~\ref{property3},
the arrival rate of class-$i$ jobs at the $\h$ system induced by $\MBS$ corresponds to the rate of dropped jobs of the M/GI/$s_i$/$s_i$ queue in Property~\ref{property3}. Therefore, the load induced by $\MBS$ on the $\h$ system $\rho_{\h}^{\MBS}$ is given by
\begin{equation*}
\rho_{\h}^{\MBS} = \frac{\lambda}{|\h|} \sum_{i=1}^C \varrho_i  E_{s_i}(\lambda \alpha_i d_i).
\end{equation*}
Finally, Proposition~\ref{prop:asi9c8s} ensures that the load at the helper set induced by $\BS$ is smaller than or equal to $\rho_{\h}^{\MBS}$.


\subsection{Proof of Theorem~\ref{th:subcritical}}

For the M/GI/$s_{i}^{(k)}$/$s_{i}^{(k)}$ queue in Property~\ref{property3},
the offered load $\lambda \frac{k}{f_k} \alpha_i d_i$
and
the number of servers
$s_{i}^{(k)}= \lfloor \frac{k \varrho_{i}}{n_i f_k\varrho} \rfloor$
both grow to infinity with the same rate, as $k\to\infty$,
and their ratio converges to $\lambda \varrho$.
\review{Within this regime, it is known that the loss probability converges to zero, i.e.,}
$$
E_{s_i^{(k)}}\left(\lambda \frac{k}{f_k} \alpha_i d_i\right) \to 0
$$
provided that
$\lambda \varrho <1$, which is the case in our setting.
Now, for the probability that in the $k$-th system an arriving job is routed to the helper set, $P_{\h}^{(k)}$,
Property~\ref{property3} and Corollary~\ref{cor:asi9c8s} imply that
$P_{\h}^{(k)}\le P_{\h}^{(k),\MBS} = \sum_{i=1}^C \alpha_i E_{s_i^{(k)}}\left(\lambda\alpha_i d_i {k}/{f_k}\right)$.
Therefore, $P_{\h}^{(k)}\to 0$ and $\RBS \to \sum_i \alpha_i d_i$ follows by \eqref{eq:ErlangB_R}.

\subsection{Proof of Theorem~\ref{th:critical}}

Recall that the M/GI/$s_i^{(k)}$/$s_i^{(k)}$ queue in Property \ref{property3}
has arrival rate $\lambda^{(k)} \alpha_i$, service times distributed as $D_i$ and
$s_{i}^{(k)} = \lfloor \psi^{(k)} \frac{k \varrho_{i}}{n_i f_k\varrho} \rfloor$.
Thus, its load is
$$\rho_i^{(k)}:= {\lambda^{(k)} \alpha_i d_i/s_{i}^{(k)}}.$$
Within the scaling \eqref{eq:HW_cond}, the following proposition states that also each M/GI/$s_i$/$s_i$ queue in Property~\ref{property3} scales in the Halfin-Whitt regime.

\begin{proposition}
\label{prop_HW_cond}
Let~\eqref{eq:cond_f_k} and~\eqref{eq:HW_cond} hold.
As $k\to\infty$,
$s_{i}^{(k)} \to\infty$,
$\psi^{(k)}\to 1$ and
\begin{align}
\label{eq:rho_i_HW}
\left( 1 - \rho_{i}^{(k)} \right) \sqrt{s_{i}^{(k)}}
\to \theta_i.
\end{align}
\end{proposition}
\begin{proof}
Since  $f_k=o(k)$ and $\psi^{(k)}\le 1$,
$s_{i}^{(k)}
=\left \lfloor \psi^{(k)}\, \frac{k}{ n_if_k} \,\frac{\varrho_{i}}{\varrho}  \right\rfloor\to\infty$.
For $x\in [0,1]$,
$k-\sum_i \lfloor x\, \frac{k}{ n_i^{(k)}} \,\frac{\varrho_{i,j}}{\varrho} \rfloor n_i^{(k)} \ge \max_i n_i^{(k)}$ if and only if
\begin{align}
\label{ascsc}
\frac{k}{f_k}-\sum_{i=1}^C  \left\lfloor x\, \frac{k}{ n_i f_k} \,\frac{\varrho_{i,j}}{\varrho}  \right\rfloor n_i  \ge \max_i n_i.
\end{align}
The RHS is constant in $k$ while the LHS grows to infinity sublinearly in $k$.
So, \eqref{ascsc} holds true for all $k$ large enough.
Using the definition of $\psi$ in Section~\ref{sec:BS}, the floor function $\lfloor\cdot\rfloor$ gives
\begin{align}
\label{eq:phi_property}
\psi^{(k)} = 1,\quad \forall k \mbox{ large enough}.
\end{align}
It remains to prove \eqref{eq:rho_i_HW}.
We obtain
\begin{align*}
1\ge \lim_{k\to\infty} \frac{1-\rho_i^{(k)}}{1-\rho^{(k)}}
& = \lim_{k\to\infty} \frac{
1 -
\lambda^{(k)} \alpha_i d_i  \lfloor \psi^{(k)} \frac{k \varrho_{i}}{n_i f_k\varrho} \rfloor^{-1}
}
{
1-\rho^{(k)}
}\\
%
& \ge \lim_{k\to\infty}\left(1 -
    \frac{
        \lambda^{(k)} \alpha_i d_i
    }{
        \psi^{(k)} \frac{k \varrho_{i}}{n_i f_k\varrho}  -1
    }
\right) \frac{1
}
{
1-\rho^{(k)}
}\\
%
%
& = \lim_{k\to\infty}
    \frac{   \psi^{(k)}  - \rho^{(k)}} { 1-\rho^{(k)}} =1
\end{align*}
where the last equality follows by \eqref{eq:phi_property}.
Therefore,
\begin{align*}
\lim_{k\to\infty}
\left( 1 - \rho_{i}^{(k)} \right) \sqrt{s_{i}^{(k)}}
& = \lim_{k\to\infty}
\left( 1 - \rho^{(k)} \right) \sqrt{\left\lfloor \psi^{(k)} \frac{k \varrho_{i}}{n_i f_k\varrho} \right\rfloor} \\
& = \lim_{k\to\infty}
\left( 1 - \rho^{(k)} \right) \sqrt{\frac{k}{f_k}}
\sqrt{\frac{\varrho_{i}}{n_i\varrho}} \\
& =
\theta  \sqrt{\frac{\varrho_{i}}{n_i\varrho}} = \theta_i
\end{align*}
which proves \eqref{eq:rho_i_HW}.
\end{proof}

Since each
M/GI/$s_i^{(k)}$/$s_i^{(k)}$ queue
scales in the Halfin-Whitt regime, we can apply Lemma~\ref{Erlang_lemma} to our case to obtain that
\begin{align*}
\frac{\phi(\theta_i)}{\Phi(\theta_i)}
 & =\lim_{k\to\infty} \sqrt{s_i^{(k)}} E_{s_i^{(k)}} (\lambda^{(k)}\alpha_i d_i)\\
 & =\sqrt{\frac{\varrho_{i}}{n_i\varrho}}\lim_{k\to\infty} \sqrt{\frac{k}{f_k}} E_{s_i^{(k)}} (\lambda^{(k)}\alpha_i d_i). \end{align*}
In view of Property~\ref{property3},
the probability that an arriving job is routed to the helper set within $\MBS$ is given by
\begin{align}
\label{eq:P_h}
P_{\h}^{(k),\MBS} = \sum_{i=1}^C \alpha_i E_{s_i}(\lambda\alpha_i d_i).
\end{align}
 Substituting in~\eqref{eq:P_h}, we obtain
\begin{align*}
\lim_{k\to\infty} \sqrt{\frac{k}{f_k}}\,P_{\h}^{(k)}
& = \sum_{i=1}^C \alpha_i \lim_{k\to\infty} \sqrt{\frac{k}{f_k}} E_{s_i}(\lambda^{(k)}\alpha_i d_i)\\
& = \theta  \sum_{i=1}^C \frac{\alpha_i}{\theta_i} \frac{\phi(\theta_i)}{\Phi(\theta_i)}.
\end{align*}
Finally, the result follows by Corollary~\ref{cor:asi9c8s}.

\section{Empirical comparison}
\label{sec:numerical}

We present the results of a numerical study that compares the simulated mean response times induced by our policy BalancedSplitting ($\BS$, Definition~\ref{def:BS}) to that of other state-of-the-art scheduling policies discussed in Section~\ref{sec:related_work}.
\review{In our test bench, we test with respect to different traffic conditions and real data traces.
We consider scenarios where, roughly speaking, jobs are either large or small, with respectively low and high probability of arriving, as this pattern is ubiquitous in real applications; see, e.g., \cite[Section~2.1]{LB_size_testing}.
}

In our simulation setup, the auxiliary policy~$\pi$ used for $\BS$ is First-Come First-Served (FCFS), and the following consideration has been taken into account.
\begin{remark}
\label{rem:po2}
We have always chosen the overall number of servers and the server needs to be all powers of two because the ServerFilling-* policies require this assumption for throughput optimality. This choice has been made to better stress the quality of the proposed policies as in this case ServerFilling-SRPT is also optimal for minimizing the mean response time in heavy traffic~\cite{Grosof22SRPT}.
\end{remark}




\subsection{Heavy-traffic and subcritical limiting regime}

In Figure~\ref{fig:1}, we have presented simulation results related to the critically loaded limiting regime.
Here, we complement those results by considering the subcritical and heavy-traffic regimes.
In the latter, we recall that the load $\rho$ approaches one from below but the overall number of servers is constant.
The results are illustrated in Figure~\ref{fig:2},
where we have used the same job classes, service time distributions and server needs used for Figure~\ref{fig:1}.
Also, each simulation run is based on $10^6$ arrivals.
\begin{figure}
\centering
\makebox[1.03\columnwidth][c]{\hspace{0cm}\includegraphics[width=1.03\columnwidth]{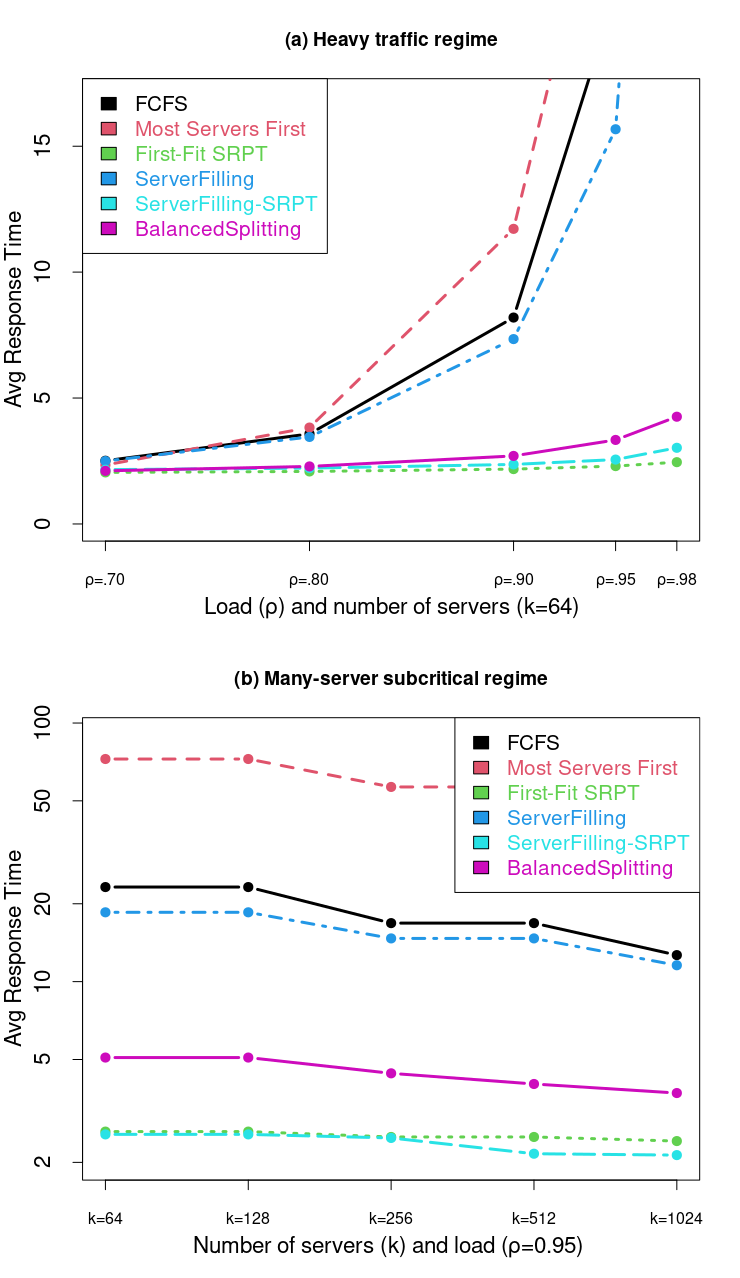}}%
%
\caption{
Simulated mean response time in the heavy-traffic (a) and subcritical (b) regimes.
Small resp. large jobs arrive with probability 0.95 resp. 0.05.
Server needs and mean service times are as in Figure~\ref{fig:1}.
Arrivals are Poisson.
}
%
%
%
%
%
%
%
%
%
%
%
%
\label{fig:2}
\end{figure}

In the heavy traffic scenario (Figure~\ref{fig:2}.a), $\BS$ performs almost optimally up until moderate loads $\rho\le 0.9$, that is up until the point where the helper set is slightly used. For higher loads, the helper set utilization becomes non-negligible and the intrinsic inefficiencies of FCFS, which is not throughput-optimal, start deteriorating performance, though the resulting response times are significantly smaller than the ones induced by FCFS, Most-Servers-First and ServerFilling.

In the subcritical scenario (Figure~\ref{fig:2}.b), $\BS$ provides again excellent performance, with average response times that are close to those achieved by ServerFilling-SRPT and First-Fit SRPT, which are preemptive and size-aware.

\subsection{Real HPC workloads}


We now assess the performance of the proposed policies against real data traces from the Parallel Workloads Archive~\cite{feitelson2014experience,rep}.
We consider the ``cleaned version'' of two datasets commonly used for HPC system benchmarks: namely,
the SDSC (San Diego Supercomputer Center) SP2 log~\cite{SDSC}
and
the KIT (Karlsruhe Institute of Technology) FH2 log~\cite{KIT}.
For each job accounted for in this dataset,
we follow the Standard Workload Format (SWF)~\cite{StandardWF}
to extract arrival time (Submit Time), service time (Run Time) and server need (Number of Allocated Processors).
%
Within these datasets, most jobs require a number of cores that is a power of two; for SDSC SP2, these jobs account for 84.4\% of all jobs.
Therefore, they are ideal to stress the performance of the proposed policy $\BS$ because within this setting, ServerFilling-SRPT minimizes the mean response time in heavy traffic~\cite{Grosof22SRPT}.
To stress the performance of our algorithm even further, we remove from our analysis those jobs whose server needs is not a power of two.
To limit the cost of simulations, we consider
jobs with a server need no larger than 64.
In our framework, this gives two models with $C=7$ classes for both datasets, whose parameters are given in Tables~\ref{tab:SDSC} and~\ref{tab:KIT}.
\begin{table}[h]
\begin{center}
\def\arraystretch{1.1}%
 \begin{tabular}{cccc}
\hline
\hline
$\mathbb{E}[D_i]$ & std($D_i$) & $n_i$ & $\alpha_i$\\
\hline
10519.71  &  18267.03  &  1  &  0.2321\\
1436.82  &  6250.19  &  2  &  0.1496\\
5643.69  &  18123.7  &  4  &  0.1624\\
9248.53  &  18468.51  &  8  &  0.1652\\
10601.46  &  17050.63  &  16  &  0.156\\
12139.59  &  22654.86  &  32  &  0.0807\\
8302.33  &  19074.81  &  64  &  0.054\\
\hline
\hline
\end{tabular}
\end{center}
\caption{Model parameters extracted from the SDSC SP2 dataset.}
\label{tab:SDSC}
\begin{center}
\def\arraystretch{1.1}%
 \begin{tabular}{cccc}
\hline
\hline
$\mathbb{E}[D_i]$ & std($D_i$) & $n_i$ & $\alpha_i$\\
\hline
1845.19  &  11440.31  &  1  &  0.7851\\
1470.13  &  5237.83  &  2  &  0.018\\
11169.87  &  38631.83  &  4  &  0.0406\\
3167.33  &  19727.29  &  8  &  0.0137\\
5706.45  &  17212.04  &  16  &  0.0539\\
60673.08  &  92531.56  &  32  &  0.0493\\
61343.42  &  106094.97  &  64  &  0.0393\\
\hline
\hline
\end{tabular}
\end{center}
\caption{Model parameters extracted from the KIT FH2 dataset.}
\label{tab:KIT}
\end{table}
Then, Figure~\ref{fig:KIT} plots the simulated response times
within both datasets and for $k=512$ and $k=1024$ servers.
%
%
As expected ServerFilling-SRPT and First-Fit SRPT always provide the best delays.
Also, all plots show that the proposed policy BalanceSplitting provides significantly better results than ServerFilling, which is preemptive, and FCFS, for almost all loads. We ought this gain to the large variability of service times across job classes and that our framework can reduce the interference between large and small jobs.
In fact, especially in Table~\ref{tab:KIT} we observe that the jobs that require 32 or 64 servers have huge service time requirements (relative to the other jobs), and within $\BS$, their execution does not interfere with the execution of smaller jobs.
\begin{figure*}
\centering
\makebox[\textwidth][c]{\hspace{0.5cm}\includegraphics[width=1.02\textwidth]{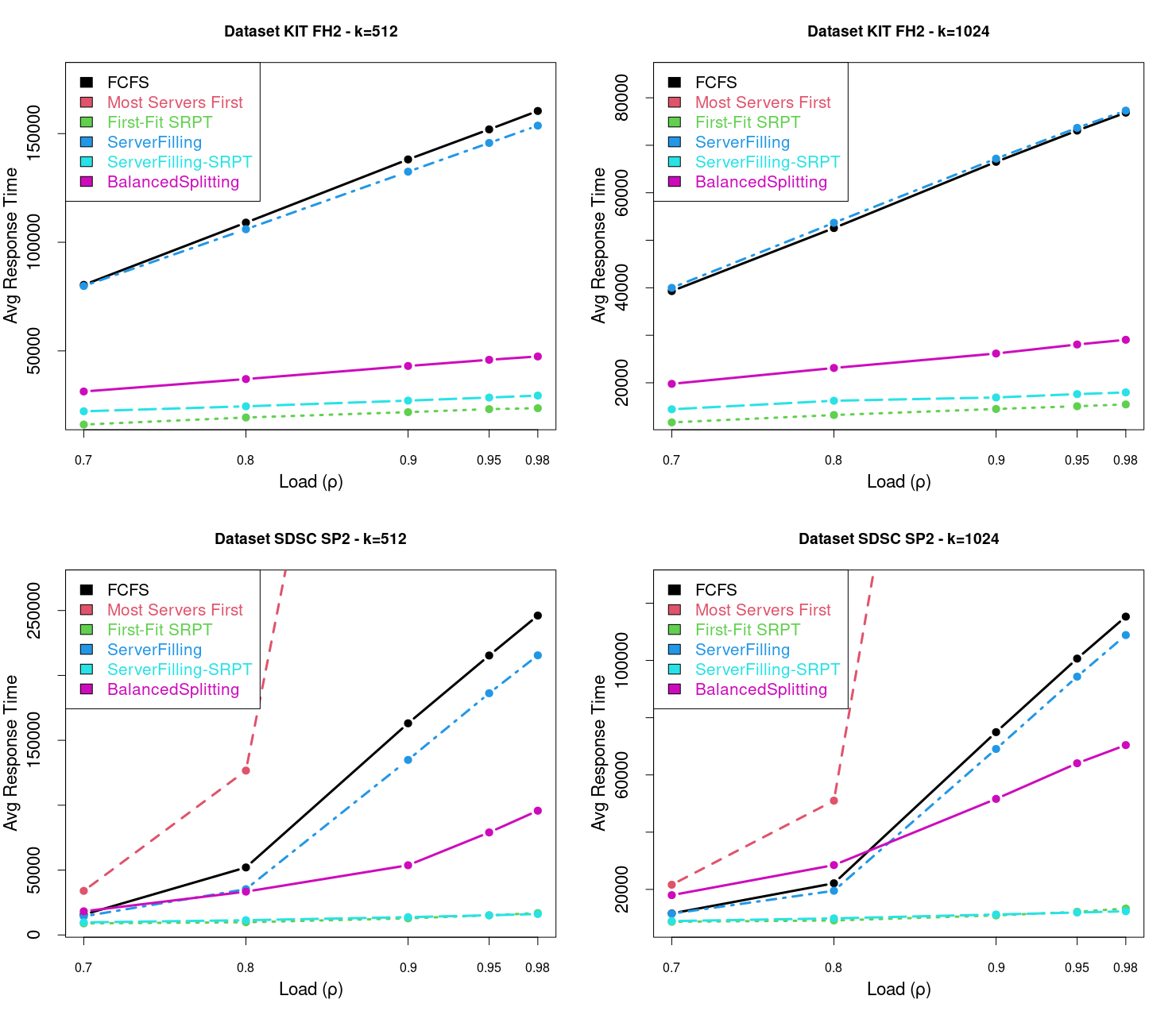}}%
\caption{Simulated response times within the KIT FH2 and SDSC SP2 datasets by varying the network load and with respect to $k=512$ (left) and $k=1024$ (right) overall servers.
The auxiliary policy $\pi$ used within BalanceSplitting ($\BS$) is FCFS.}
%
%
%
%
\label{fig:KIT}
\end{figure*}

%


%

\section{Conclusion}

We have proposed Balanced-Splitting, a nonpreemptive and job-size oblivious scheduling policy for multiserver jobs that operate on a static and balanced partitioning of the server set.
Our main results show that the proposed policy has the zero-wait property asymptotically. More precisely,
the probability that a randomly arriving job needs to wait for service vanishes when the system size grows to infinity. Numerical simulations have shown that Balanced-Splitting provides much better delay performance than other common scheduling policies such as FCFS. Also, its average response time is competitive with the one induced by state-of-the-art preemptive and size-aware policies.

We have assumed that the static and balanced partitioning of the server set only depends on the class of each job.
In our work, the class of a job is essentially determined by its server need, which is reasonable because this information is available precisely. To define the server set partitioning, we have assumed the per-class knowledge of the probability distribution of service times, which in practice can be obtained by user profiling techniques.
We believe that a better partitioning of the server set may be obtained in a size-aware setting. Here, a class of a job would be determined by its server need \emph{and} whether its service time belongs to a specified interval.
Of course, the scheduler needs to know the size of each incoming job to properly assign to the proper subset of servers and, in contrast to the approach followed in this work, this would make any policy size-ware.
On the other hand, we believe that a balanced size-aware partitioning yields better delay performance. We leave this question as future research.



\ifCLASSOPTIONcaptionsoff
  \newpage
\fi



\bibliographystyle{IEEEtran}
%



%


\begin{IEEEbiography}[{\includegraphics[width=1in,height=1.25in,clip,keepaspectratio]{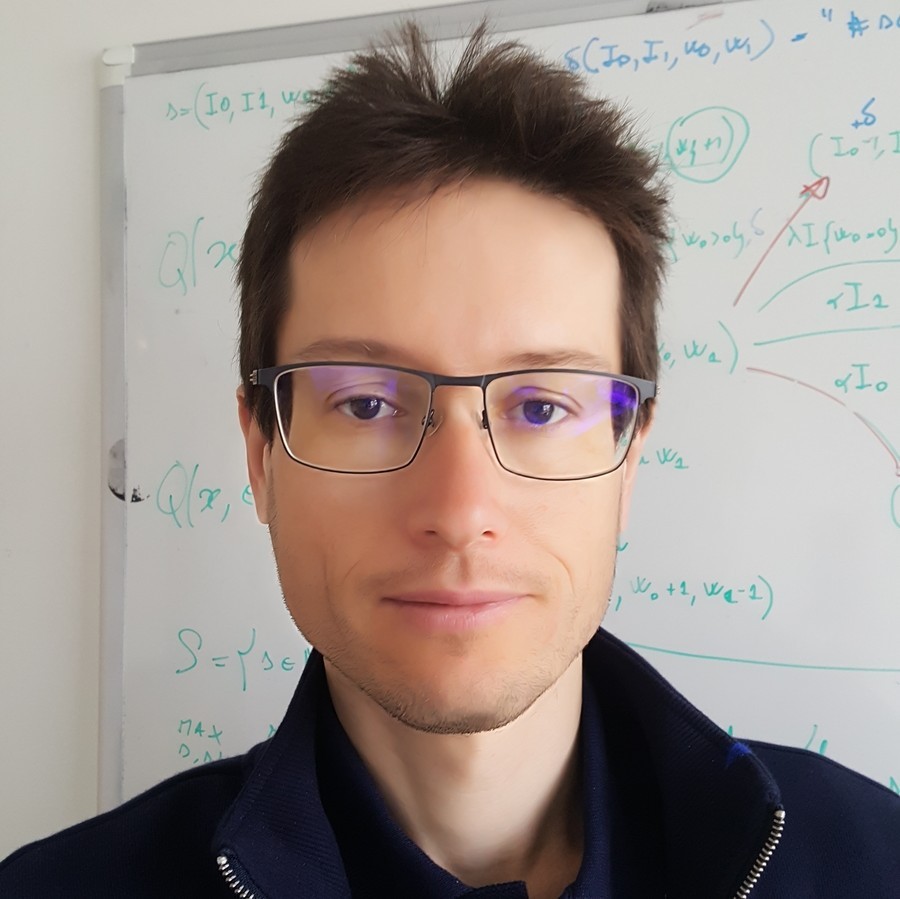}}]{Jonatha Anselmi}
is a tenured researcher at the French National Institute for Research in Digital Science and Technology (Inria), since 2014. Prior
to this, he was a full-time researcher at the Basque Center for Applied
Mathematics and a postdoctoral researcher at Inria. He received his PhD
in computer engineering at Politecnico di Milano (Italy) in 2009.
At the intersection of applied mathematics, computer science and engineering, his research interests focus on decision making under uncertainty, with particular emphasis on the development of highly-scalable algorithms that minimize congestion and operational costs of large-scale distributed systems.
\end{IEEEbiography}

\begin{IEEEbiography}[{\includegraphics[width=1in,height=1.25in,clip,keepaspectratio]{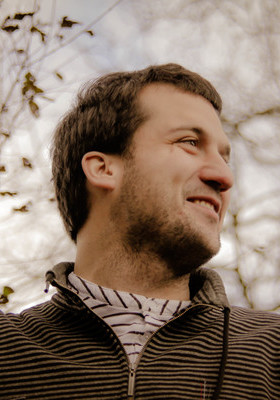}}]{Josu Doncel}
is an associate professor in the University of the Basque Country. He obtained from the same university the Industrial Engineering degree in 2007, the Mathematics degree in 2010 and, in 2011, the Master degree in Applied Maths and Statistics. He received in 2015 the PhD degree from Universit\'e de Toulouse (France). He has previously held research positions at LAAS-CNRS (France), INRIA Grenoble (France) and BCAM-Basque Center for Applied Mathematics (Spain), teaching positions at ENSIMAG (France), INSA-Toulouse (France) and IUT-Blagnac (France) and invited professor positions at CentraleSupelec (France), Inria Paris (France) and David laboratory (France).
\end{IEEEbiography}



\vfill


\end{document}